\documentclass[conference]{IEEEtran}
\IEEEoverridecommandlockouts
\usepackage{cite}

\usepackage{amsmath,amssymb,amsfonts}

\usepackage{graphicx}
\usepackage{textcomp}
\usepackage{xcolor}
\usepackage{bm}
\usepackage{amsthm}
\usepackage{mathabx}
\usepackage{changepage}
\usepackage{mathtools}
\usepackage{pgfplots}
\usepackage{graphics}
\usepackage{pifont}
\usepackage{subfigure}
\usepackage{arydshln}
\usepackage{empheq}
\usepackage{romannum}
\usepackage{algorithm}
\usepackage{algpseudocode}
\usepackage{scalerel}
\usepackage{enumerate}
\usepackage{color}
\usepackage{hyperref}
\usepackage{MnSymbol}
\usepackage[capitalize]{cleveref}
\usepackage{array}
\usetikzlibrary{spy}

\newtheorem{prop}{Proposition}

\newcommand{\bl}[1]{\textcolor[rgb]{0,0,0}{#1}}
\newcommand{\rd}[1]{\textcolor[rgb]{0,0,0}{#1}}
\newcommand{\pp}[1]{\textcolor[rgb]{0,0,0}{#1}}
\newcommand{\ff}[1]{\textcolor[rgb]{0,0,0}{#1}}
\newcommand{\blu}[1]{\textcolor[RGB]{0,0,0}{#1}}

\usepackage[inline]{enumitem}

\def\BibTeX{{\rm B\kern-.05em{\sc i\kern-.025em b}\kern-.08em
    T\kern-.1667em\lower.7ex\hbox{E}\kern-.125emX}}
\begin{document}

\title{Sparse Spatial Smoothing:
Reduced Complexity and Improved Beamforming Gain via Sparse Sub-Arrays\\
}

\author{\IEEEauthorblockN{Yinyan Bu$^\ast$, Robin Rajam\"{a}ki$^\ast$, Anand Dabak$^\ddagger$, Rajan Narasimha$^\ddagger$, Anil Mani$^\ddagger$, Piya Pal$^\ast$}
\IEEEauthorblockA{$^\ast$\textit{Department of Electrical and Computer Engineering, University of California San Diego}, CA, USA\\
$^\ddagger$\textit{Texas Instruments}, Dallas, TX, USA}


}

\maketitle

\begin{abstract}
\bl{This paper} addresses the problem of single snapshot Direction-of-Arrival (DOA) estimation, \bl{which is of great importance in a wide-range of applications including automotive radar.} A popular approach \bl{to achieving high angular resolution when only one temporal snapshot is available} is \bl{via subspace methods using} spatial smoothing. This involves \bl{leveraging spatial shift-invariance in the antenna array geometry---typically a uniform linear array (ULA)---to} rearrange the single snapshot measurement vector into a spatially smoothed matrix \bl{that reveals the signal subspace of interest}. However, \bl{conventional approaches using spatially shifted ULA sub-arrays} can lead to \bl{a prohibitively high} computational complexity \bl{due to the large dimensions of the resulting spatially smoothed matrix}. \bl{Hence, we} propose \bl{to instead employ judiciously designed} sparse \bl{sub-arrays, such as nested arrays,} to reduce \bl{the} computational complexity of spatial smoothing while retaining the aperture and identifiability \bl{of conventional ULA-based approaches}. \bl{Interestingly, this idea also suggests a} novel beamforming \bl{method which linearly combines multiple spatially smoothed matrices corresponding to different sets of shifts of the sparse (nested) sub-array.} \bl{This so-called \emph{shift-domain beamforming} method is demonstrated} to boost the effective SNR, \bl{and thereby resolution, in a desired angular} region of interest, \bl{enabling single snapshot} low-complexity DOA estimation with identifiability guarantees.
\end{abstract}


\section{Introduction}

Direction-of-Arrival (DOA) \bl{estimation using multi-antenna arrays} is essential in \bl{both} wireless communications \bl{and} radar. However, in \bl{many} applications such as automotive radar \cite{patole2017automotive} \bl{or} joint communication and sensing\cite{zhang2021overview}, \bl{signal} sources may be coherent \bl{due to multi-path} and the environment highly dynamic due to the mobility of \bl{the sources or radar targets}. As a result, the number of available \bl{temporal} snapshots in a given coherence interval is limited; in the worst case, only single snapshot is available \cite{patole2017automotive,sun20214d}. In the single-snapshot scenario, a common practice is \bl{employing} a uniform linear array (ULA), whose inherent shift-invariant structure can be leveraged to \bl{identify the DOAs from} the signal subspace. In particular, spatial smoothing \cite{odendaal1994two,pillai1989forward} \bl{is typically} used to \bl{reshape the single-snapshot measurement vector of the ULA into a structured matrix of appropriate rank, from which the} DOAs can then \bl{be estimated using} high-resolution subspace methods, such as MUSIC \cite{schmidt1986multiple} or ESPRIT \cite{roy1989esprit}. 

On the one hand, in some applications like automotive radar, the number of targets that need to be resolved in the angular domain of any given range-Doppler bin can be small \cite{sun2020mimo}, therefore trading off identifiability for resolution is of great interest, where achieving superior resolution at the expense of identifiability may be \bl{acceptable.} 
On the other hand, \bl{in addition to high angular resolution,} many applications \bl{also} require \emph{fast} DOA estimation. For example, current automotive radars provide a whole radar cube---range, Doppler, and angle---every $50$ ms \cite{roldan2023low}. Moreover, DOA estimation \bl{is typically} implemented in embedded hardware with \bl{minimal} computational resources. \bl{As the computational complexity of DOA estimation typically grows with the number of antennas, computational resources may become a limiting factor when deploying increasingly larger arrays---}as is the trend \bl{at, for example, millimeter-wave wavelengths}. For these reasons, low-complexity DOA estimation algorithms are needed. 

\bl{In subspace methods, the main computational burden consists of computing the singular value decomposition (SVD) of the (possibly spatially smoothed) measurement matrix. However, in conventional spatial smoothing using ULA sub-arrays, the dimensions of this matrix are typically proportional to the number of antennas. Hence, the complexity of SVD may be prohibitively large in real-time applications employing many antennas.} 
\pp{The number of antennas could be reduced by employing a sparse array geometry \cite{wang2017coarrays}, which in addition to lowering hardware costs, can enhance resolution \cite{sarangi2023superresolution} and reduce mutual coupling \cite{liu2016supernested}. However, compared to ULAs, which accommodate many choices of the sub-array used for spatial smoothing, designing sparse array geometries that are both amenable to spatial smoothing (by virtue of a shift-invariant structure) while providing identifiability guarantees is more challenging.} 
This raises \bl{the following} question: ``\bl{Can} single-snapshot \bl{high-resolution} DOA estimation be achieved \bl{at} low \bl{computational} complexity?''

\textbf{Contributions:} 
\bl{This paper answers the above question in the positive by utilizing \emph{sparse \pp{sub-}arrays} for spatial smoothing. We show that the computational complexity of single-snapshot DOA estimation using subspace methods \pp{(implemented on general-purpose solvers)} can be significantly reduced---without compromising resolution---by using sparse (rather than ULA) sub-arrays for spatial smoothing.
Specifically, given a single-snapshot measurement vector from an $N$-antenna ULA, we show that the complexity of subspace methods employing spatial smoothing can be reduced by a factor of $\sqrt{N}$ 
using judiciously designed sparse sub-array geometries, 
such as nested arrays \cite{pal2010nested} while \pp{provably identifying} $\mathcal{O}(\sqrt{N})$ signal sources/targets.} 
\bl{This ``sparse spatial smoothing $S^3$'' approach retains a sub-array aperture proportional to that of the full ULA while utilizing all $N$ measurements.} Furthermore, we \bl{demonstrate that effective SNR, and thereby resolution, can actually be improved in a desired angular region of interest via a} novel shift-domain beamforming method. \bl{Simulations} validate that sparse array-based spatial smoothing using shift-domain beamforming \bl{can offer improved resolution} at reduced \bl{computational} complexity compared to conventional spatial smoothing using ULA sub-arrays.

\textbf{Notation:} Given an array geometry $\mathbb{S}=\{d_1,d_2,\cdots,d_M\}$, matrix $\mathbf{A}_{\mathbb{S}}(\bm{\theta}) \in \mathbb{C}^{M\times K}$ denotes the array manifold matrix for sensors located at $n\lambda/2$, where $n\in \mathbb{S}$, $\bm{\theta}\in[-\pi/2,\pi/2)^K$ denote the target directions, and $[\mathbf{A}_{\mathbb{S}}(\bm{\theta})]_{m,n}=\exp(j\pi d_m \sin\theta_n)$. 
Furthermore, the aperture of $\mathbb{S}$ is denoted by $Aper(\mathbb{S})$. The largest eigenvalue of $\mathbf{B}$ is denoted by $\lambda_{max}(\mathbf{B})$ and $\mathbf{u}_1(\mathbf{B})$ denotes the eigenvector corresponding to $\lambda_{\max}(\mathbf{B})$. \bl{Moreover,} $vec(\mathbf{V})$ denotes the vectorized matrix $\mathbf{V}$ \blu{and} $w^*$ denotes the conjugate of $w\in\mathbb{C}$.

\section{Measurement Model \& Generalized Spatial Smoothing}
\subsection{Measurement Model}
Consider $K$ narrowband source signals impinging on an $N$-sensor ULA $\mathbb{S}=\{0,1,\cdots,N-1\}$ from distinct directions $\bm{\theta}=[\theta_1, \theta_2, \cdots, \theta_K]^T$. The single snapshot received signal $\mathbf{y}\in\mathbb{C}^{N\times 1}$ is of the following form: 
\begin{equation}\label{eqn:full_measurement}
\mathbf{y}=\bm{A}_{\mathbb{S}}(\bm{\theta})\bm{x} + \bm{n} ,
\end{equation}
where $\bm{x}\in\mathbb{C}^K$ is the source/target signal and $\bm{n}\in\mathbb{C}^{N}$ is a noise vector. Note that \eqref{eqn:full_measurement} is applicable to both passive and active sensing---indeed, $\mathbb{S}$ can represent either a physical or virtual array (\blu{i.e., $N$ can be the number of physical or virtual sensors}). One prominent virtual array model is the sum co-array in active sensing (e.g., co-located MIMO radar\cite{li2007mimo}). The goal is to estimate $\{\theta_k\}_{k=1}^K$ from $\mathbf{y}$.

\subsection{Generalized Spatial Smoothing}
Spatial smoothing \cite{shan1985spatial} is a widely used technique for single snapshot DOA estimation, where the inherent ``shift invariant" structure of \rd{certain array geometries---conventionally, ULAs---}is leveraged to identify the ``signal subspace'' of interest from $\mathbf{y}$. 
Specifically, 
\rd{given a so-called basic sub-array $\mathbb{S}_b=\{d_1,d_2,\cdots,d_{N_b}\}\subseteq\mathbb{S}$ and a set of $Q\in\mathbb{N}_+$ shifted copies of $\mathbb{S}_b$, each in $\mathbb{S}$, measurement vector}
$\mathbf{y}$ is partitioned into vectors $\bm{y}_1,\bm{y}_2,\cdots,\bm{y}_Q\in\mathbb{C}^{N_b}$ 
\rd{corresponding to these shifted sub-arrays $\mathbb{S}_i=\mathbb{S}_b+\delta_i\subseteq \mathbb{S}$, where} $\delta_i\in\mathbb{Z}, i=1,2,\ldots,Q$. 
\rd{These vectors} 
are then rearranged into an augmented measurement matrix:
\begin{equation}\label{eqn:rearrange}
    \mathbf{Y}=[\bm{y}_1,\bm{y}_2,\cdots,\bm{y}_Q].
\end{equation}
In absence of noise \rd{($\bm{n}=\bm{0}$)}, \rd{$\bm{y}_i$ can be written as} $\bm{y}_i=\bm{z}_i=\bm{A}_{\mathbb{S}_b}(\bm{\theta})\bm{D}_i(\bm{\theta})\bm{x}$, where $\bm{D}_i(\bm{\theta})\in\mathbb{C}^{K\times K}$ is a diagonal matrix with $[\bm{D}_i(\bm{\theta})]_{k,k}=\exp(j\pi\delta_i\sin\theta_k)$. 
The augmented matrix $\mathbf{Y}$ \rd{therefore} permits the following decomposition \cite{bu2023harnessing}:
\begin{align}
\mathbf{Y}&=\bm{A}_{\mathbb{S}_b}(\bm{\theta})[\bm{D}_1\bm{x},\bm{D}_2(\bm{\theta})\bm{x},\cdots,\bm{D}_Q(\bm{\theta})\bm{x}]\nonumber\\
&=\bm{A}_{\mathbb{S}_b}(\bm{\theta})\text{diag}(\bm{x})\bm{A}_{\Delta}(\bm{\theta})^{T}. \label{eqn:ss_op}
\end{align}
where $\Delta=\{\delta_1,\delta_2,\cdots,\delta_Q\}$ represents the set of \rd{(integer)} shifts.

In conventional ULA-based spatial smoothing, $\mathbb{S}_b$ and $\Delta$ are both ULAs with $N_b$ and $Q$ satisfying
\begin{equation}\label{eqn:N_b and Q}
    N_b+Q-1=N.
\end{equation}
It is well-known that varying $N_b$ \rd{or} $Q$ \rd{for a given $N$} allows trading off aperture (resolution) for identifiability. A typical choice is $N_b=\lceil\frac{N}{2}\rceil$, which maximizes the number of identifiable sources \cite{shan1985spatial}.

\cref{eqn:ss_op} illustrates that the shift-invariant structure of $\mathbb{S}$ enables building the rank of $\mathbf{Y}$ on which subspace methods can be applied to identify $\bm{\theta}$. In our recent work \cite{bu2023harnessing}, we investigated generalized spatial smoothing for arbitrary (sparse) $\mathbb{S}$, 
$\mathbb{S}_b$ and 
$\Delta$ showing when exact recovery is possible. 
\pp{A special case of interest is when $\mathbb{S}_b$ and $\Delta$ contain ULA segments. Indeed,}
it is well-known that \pp{the array manifold matrix of an array with 
a ULA segment of length $K$ guarantees that the corresponding manifold matrix has a Kruskal rank of at least $K$ due to containing a Vandermonde sub-array of appropriate size}. 
Therefore, when $\mathbb{S}_b$ contains a ULA segment of length $K+1$ and $\Delta$ contains a ULA segment of length $K$, it \pp{can be shown that} that MUSIC can identify $K$ distinct angles from $\mathbf{Y}$ \pp{\cite{bu2023harnessing}}. 
Moreover, \pp{note that when} $\mathbb{S}$ is a ULA, \pp{as is assumed in this paper, many choices for $\mathbb{S}_b$ and $\Delta$ beyond ULAs are possible}. \pp{Hence,} sparse array geometries \pp{with a ULA segment of appropriate length (such as nested arrays)} can be used for $\mathbb{S}_b$ in spatial smoothing \bl{to increase aperture, and thereby resolution, while guaranteeing the identifiability of a desired number of targets}.

\subsection{Problem Setting}
For applying MUSIC, the singular value decomposition (SVD) of $\mathbf{Y}$ \bl{is computed} to find the corresponding ``noise subspace" with complexity $\mathcal{O}(\min(N_b^2Q,N_bQ^2))$ \cite{golub2013matrix}. However, for conventional ULA-based spatial smoothing, $N_b$ and $Q$ are usually $\propto N$, which will lead to SVD having a complexity of $\mathcal{O}(N^3)$. In this paper, we \bl{explore spatial smoothing schemes that offer high-resolution DOA estimation at low computational complexity. This may be of great interest in applications such as automotive radar.} In particular, \bl{we ask:}
\begin{enumerate}[label=Q$_\arabic*$]
    \item Can we reduce the complexity of SVD by an alternative choice of $\mathbb{S}_b$ while ensuring: ($\romannumeral1$) identifiability of $K$ targets and ($\romannumeral2$) an aperture \bl{proportional to $N$}?\label{i:Q1}
\end{enumerate}
 \bl{Specifically, we investigate whether sparse arrays, such as nested arrays \cite{pal2010nested}, can answer this question in the positive. Indeed, judicious sparse array designs are known} to achieve an aperture \bl{on the order} of $N$ using only $\sqrt{N}$ sensors. 
 \bl{By virtue of the ULA segment in the nested array (also with on the order of $\sqrt{N}$ sensors), identifiability of $K=\mathcal{O}(\sqrt{N})$ targets via spatial smoothing MUSIC can be guaranteed when shift set $\Delta$ is a ULA with on the order of at least $K$ sensors.} 
 However, the employed \bl{sparse} sampling scheme should also use all \bl{$N$} independent measurements in \bl{\eqref{eqn:full_measurement}, as doing} otherwise would be wasteful in presence of noise. Hence, we modify \labelcref{i:Q1} as follows:
\begin{enumerate}[label=Q$_\arabic*$,resume]
\item Can we achieve the goals in \labelcref{i:Q1} while utilizing all $N$ measurements in \eqref{eqn:full_measurement}\bl{---increasing SNR \ff{(potentially over a region of interest)}---without increasing complexity?}
\end{enumerate}

In the next section, we provide a positive answer to this question by \bl{showing how to leverage sparse arrays to reduce the computational complexity of subspace methods employing spatial smoothing.} \bl{We also present a} novel \bl{shift-domain beamforming approach for spatial smoothing, where multiple} sparse spatially smoothed measurement matrices \bl{are linearly combined} to improve SNR and angular resolution \bl{in a given region of interest}.

\section{Sparse Spatial Smoothing $S^3$: Reduced computational complexity \bl{\& improved SNR}}

To illustrate \bl{the main idea of sparse spatial smoothing and shift-domain beamforming, suppose} $\mathbb{S}_b$ is a nested array \cite{pal2010nested} with $N_b=2r(N)-2=\mathcal{O}(\sqrt{N})$ sensors, where $r(N)\triangleq\lfloor\sqrt{N}\rfloor$. That is, 
\begin{equation}\label{eqn:nested_array_geom}
    \mathbb{S}_b=\{0,1,\cdots,r(N)-2\}\bigcup\{m(r(N)-1)-1\}_{m=2}^{r(N)},
\end{equation}
with $Aper(\mathbb{S}_b)=r(N)(r(N)-1)-1=N-\sqrt{N}+C_1\bl{\propto} N$ where $C_1$ is a constant. 

The total number of shifts is given by $N-Aper(\mathbb{S}_b)+1$, and we divide them into $L$ overlapping subsets with each subset of cardinality $P$ as follows:
\begin{equation}\label{eqn:sequence of Delta}
\begin{aligned}
    \Delta_1&=\{0,1,\cdots,P-1\}\\
    \Delta_2&=\{1,2,\cdots,P\}\\
    &\vdots\\
    \Delta_L&=\{L-1,L,\cdots,L+P-2\}.
\end{aligned}
\end{equation}
Note that in this case $P$ and $L$ satisfy:
\begin{equation}
    L+P=N-Aper(\mathbb{S}_b)+1,
\end{equation}
which yields a trade-off between identifiability and complexity. In this paper, we choose $P=r(N)-C_2\bl{\propto}\sqrt{N}$ where $C_2$ is a constant ($C_2\leq 2$ implies the maximum identifiability is $r(N)-2$). This leads to $L=\mathcal{O}(1)$.


The spatially smoothed measurement matrix $\mathbf{Y}_l$ ($l=1,\cdots,L$) corresponding to $\mathbb{S}_b$ and $\Delta_l$ can be represented as (when $\mathbf{n}=\mathbf{0}$):
\begin{align}
    \mathbf{Y}_l&=[\bm{z}_l,\bm{z}_{l+1},\cdots,\bm{z}_{l+P-1}]\nonumber\\
    &=\bm{A}_{\mathbb{S}_b}(\bm{\theta})\text{diag}(\bm{x})\bm{A}_{\Delta_l}(\bm{\theta})^{T}\nonumber\\
    &=\bm{A}_{\mathbb{S}_b}(\bm{\theta})\text{diag}(\bm{x})\bm{C}_l(\bm{\theta})\bm{A}_{\Delta_1}(\bm{\theta})^{T},\label{eq:Yl}
\end{align}
where the $p$-th column of $\mathbf{Y}_l$ is $\bm{z}^{(s)}_{l+p-1}$, which contains the elements of $\mathbf{y}$ corresponding to sub-array $\mathbb{S}_b+(p+l-2)$.
$\bm{C}_l(\bm{\theta})\in\mathbb{C}^{K\times K}$ is a diagonal matrix with $[\bm{C}_l(\bm{\theta})]_{k,k}=\exp(j\pi(l-1)\sin\theta_k)$. \bl{Hence, matrix} $\bm{C}_l(\bm{\theta})$ captures the relative shifts between the elements in $\Delta_l$ and $\Delta_1$.

Note that $\mathbf{Y}_l$ is a $(2r(N)-2)\times(r(N)-C_2)$ matrix. Hence $\mathcal{O}(\sqrt{N})$ targets can be identified by subspace methods applied to $\mathbf{Y}_l$ due to the ULA segment in $\mathbb{S}_b$ and $\Delta_l$.
As we will show later, the complexity of performing SVD on \bl{spatially smoothed matrix} $\mathbf{Y}_l$ \bl{(corresponding to a nested sub-array)} is $\mathcal{O}(N^{\frac{3}{2}})$, which is significantly smaller than \bl{the} complexity $\mathcal{O}(N^2)$ of performing conventional ULA-based spatial smoothing on 
a ULA \bl{sub-array} of equivalent aperture. 

For appropriate choice of $P$, we can \bl{naturally} utilize all $N$ measurements with $L=1$. \bl{However,} are there any benefits of choosing $L>1$ in the presence of noise? \bl{Next, we show that the answer is yes: we can improve the effective SNR in a given angular region of interest by taking a suitable weighted sum of $\{\mathbf{Y}_l\}_{l=1}^L$. This corresponds to a novel form of beamforming which we denote as ``shift-domain'' beamforming.}

\subsection{Shift-domain beamforming}
\bl{Given complex-valued weights $\{w_l\}_{l=1}^L\subset\mathbb{C}$,} define the weighted linear combination of all $\mathbf{Y}_l$ as $\mathbf{\Bar{Y}}:=\sum_{l=1}^L\mathbf{Y}_l\cdot w_l^*$. \bl{Hence, by \eqref{eq:Yl}}, $\mathbf{\Bar{Y}}$ can be rewritten as follows (when $\mathbf{n}=\mathbf{0}$):
\begin{align}
    \mathbf{\Bar{Y}}
    &=\sum_{i=1}^L\bm{A}_{\mathbb{S}_b}(\bm{\theta})\text{diag}(\bm{x})\bm{C}_i(\bm{\theta})\bm{A}_{\Delta_1}(\bm{\theta})^{T}\cdot w_i^*\nonumber\\
    &=\bm{A}_{\mathbb{S}_b}(\bm{\theta})\text{diag}(\bm{x})\underbrace{(\sum_{i=1}^L\bm{C}_i(\bm{\theta})\cdot w_i^*)}_{\triangleq\bm{\Bar{C}}(\bm{\theta})}\bm{A}_{\Delta_1}(\bm{\theta})^{T}.\label{eq:Y_bar}
\end{align}
Here, \bl{$\bm{\Bar{C}}(\bm{\theta})$ is a $K\times K$ diagonal matrix}
\begin{align*}
    \bm{\Bar{C}}(\bm{\theta})&=\begin{bmatrix}
        B(\theta_1) & 0 & \cdots & 0 \\
        0 & B(\theta_2) & \cdots & 0 \\
        \vdots & \vdots & \vdots & \vdots \\
        0 & 0 & \cdots & B(\theta_K)\\
    \end{bmatrix},
\end{align*}
where $B(\theta_k)=\sum_{l=1}^Lw_l^*e^{j\pi(l-1)\sin\theta_k}$ denotes the \emph{beam gain} 
in direction $\theta_k$.

\subsection{Reduction in Computational Complexity}

\bl{Next} we \pp{discuss} the \pp{computational} complexity of the proposed sparse spatial smoothing (``$S^3$") approach and conventional spatial smoothing \pp{approaches using ULA sub-arrays of various sizes. For simplicity, we consider general-purpose solvers applied to the spatially smoothed matrices $\bar{\mathbf{Y}}$ in \eqref{eq:Y_bar} and $\mathbf{Y}$ in \eqref{eqn:rearrange}. Specialized solvers utilizing the structure of said matrices to reduce complexity is left for future work}. 

\pp{Consider $S^3$ with the nested sub-array} $\mathbb{S}_b$ as in \eqref{eqn:nested_array_geom} with $N_b=\pp{\mathcal{O}(\sqrt{N})}$, $P=\pp{\mathcal{O}(\sqrt{N})}$, and $L=\mathcal{O}(1)$. 
Constructing the corresponding spatially smoothed measurement matrix $\mathbf{\Bar{Y}}\in\mathbb{C}^{N_b\times P}$ in \eqref{eq:Y_bar} requires $LN_bP$ multiplications / additions (\blu{floating point operations}), which corresponds to a complexity of $\mathcal{O}(N)$ . Since the dimensions of $\bar{\mathbf{Y}}$ are on the order of $\mathcal{O}(\sqrt{N})\times \mathcal{O}(\sqrt{N})$, \pp{the dominant computational cost of subspace methods (using general-purpose solvers) is computing the SVD of $\mathbf{\Bar{Y}}$, which incurs} complexity $\mathcal{O}(N_bP^2)=\mathcal{O}(N^{\frac{3}{2}})$.

\pp{In case of ULA-based spatial smoothing, the basic sub-array $\tilde{\mathbb{S}}_{b}$ is a $\tilde{N}_b$ sensor ULA. 
We consider three choices of $\tilde{N}_b$: the 
\begin{enumerate*}[label=(\roman*)]
    \item conventional identifiability-maximizing choice, $\tilde{N}_b=N/2+1$ (assuming $N$ is even for simplicity); \label{i:ula_N_over_2}
    \item same sub-array aperture as $S^3$, i.e., $\tilde{N}_b=Aper(\mathbb{S}_b)=\mathcal{O}(N)$ and;\label{i:ula_aper}
    \item same number of sub-array sensors as $S^3$, i.e., $\tilde{N}_b=N_b=\mathcal{O}(\sqrt{N})$.\label{i:ula_sensors}
\end{enumerate*}
By \eqref{eqn:N_b and Q}, the corresponding number of sub-array shifts $Q$ 
is \labelcref{i:ula_N_over_2} $Q=N/2$, \labelcref{i:ula_aper} $Q=N-Aper(\mathbb{S}_b)+1=\mathcal{O}(\sqrt{N})$, and \labelcref{i:ula_sensors} $Q=N-N_b+1=\mathcal{O}(N)$. The dimensions of the spatially smoothed matrix $\mathbf{Y}\in\mathbb{C}^{\tilde{N}_b\times Q}$ are then on the of order of \labelcref{i:ula_N_over_2} $\mathcal{O}(N)\times \mathcal{O}(N)$, \labelcref{i:ula_aper} $\mathcal{O}(N)\times \mathcal{O}(\sqrt{N})$, and \labelcref{i:ula_sensors} $ \mathcal{O}(\sqrt{N})\times \mathcal{O}(N)$. Hence, computing the SVD of $\mathbf{Y}$ using a general-purpose solver incurs complexity \labelcref{i:ula_N_over_2} $\mathcal{O}(N^3)$, \labelcref{i:ula_aper} $\mathcal{O}(N^2)$, and \labelcref{i:ula_sensors} $\mathcal{O}(N^2)$.
}

\pp{The above discussion shows that} the computational complexity of \pp{subspace methods employing spatial smoothing (and general-purpose solvers)} can be reduced by a factor of \bl{$\sqrt{N}$} \pp{or more using $S^3$ instead of} conventional ULA sub-array based approaches. \cref{sec:sim} will demonstrate this reduction in computational complexity via numerical simulations.

\subsection{Design of \bl{shift-domain beamforming weights} to boost SNR over \bl{angular} region of interest}
\bl{In presence of noise}, shift-domain beamforming also affects the corresponding noise matrix. In this case, each $\mathbf{Y}_l$ contains an additional spatially smoothed noise matrix $\mathbf{N}_l\in\mathbb{C}^{N_b\times P}$ (similarly, $p$-th column of $\mathbf{N}_l$ contains the elements of $\bm{n}$ corresponding to sub-array $\mathbb{S}_b+(p+l-2)$). Therefore we have:
\begin{equation}
    \begin{aligned}
        \mathbf{\Bar{Y}}&=\bm{A}_{\mathbb{S}_b}(\bm{\theta})\text{diag}(\bm{x})\bm{\Bar{C}}(\bm{\theta})\bm{A}_{\Delta_1}(\bm{\theta})^{T}+\sum_{l=1}^L\mathbf{N}_l\cdot w_l^*\\
        &=\bm{A}_{\mathbb{S}_b}(\bm{\theta})\text{diag}(\bm{x})\bm{\Bar{C}}(\bm{\theta})\bm{A}_{\Delta_1}(\bm{\theta})^{T}+\mathbf{\Bar{N}}(\mathbf{w})
    \end{aligned}
\end{equation}
where $\mathbf{w}=[w_1,w_2,\cdots,w_L]^T\in\mathbb{C}^{L\times 1}$ the shift-domain beamforming weight vector. \bl{The question then naturally emerges how to select shift-domain beamforming weights $\{w_l\}_{l=1}^L$?} In \bl{the following, we optimize these weights to maximize SNR over an angular region of interest.} Denote $\mathbf{a}(\theta)=[1,e^{j\pi\sin\theta},\cdots,e^{j\pi(L-1)\sin\theta}]^T$, such that $B(\theta)=\mathbf{w}^H\mathbf{a}(\theta)$. The effective SNR for a single target with \bl{(deterministic) amplitude $x$ and DOA} $\theta$ is then defined as follows:
\begin{equation}
    \text{SNR}(\theta)=\frac{\lVert x\mathbf{a}_{\mathbb{S}_b}(\theta)\mathbf{w}^H\mathbf{a}(\theta)\mathbf{a}_{\Delta_1}^T(\theta)\rVert_F^2}{\mathbb{E}(\lVert\mathbf{\Bar{N}}(\mathbf{w})\rVert_F^2)}.
\end{equation}
Suppose we have prior knowledge of an \bl{angular} region of interest $\mathbf{\Theta}=(\theta_l,\theta_h)$ containing \bl{the unknown DOAs} $\{\theta_k\}_{k=1}^K$ ($\theta_{l}<\theta_k<\theta_{h}\forall k$). 
This leads to an interesting question: in presence of noise, given a field of view $(\theta_l,\theta_h)$, how to design $\mathbf{w}$ so that the \bl{average} SNR over $(\theta_l,\theta_h)$ is maximized? For simplicity, consider the reduced angle $u=\pi\sin\theta$ and the corresponding region of interest $\mathbf{U}=(u_l,u_h)=(\pi\sin\theta_l,\pi\sin\theta_h)$. The optimization problem can be formulated as below:
\begin{equation}\label{eqn:opt}
    \mathbf{w}_o=\mathop{\arg\max}_{\mathbf{w}\in\mathbb{C}^{L}}\int_{\mathbf{U}}\text{SNR}(u) \,du.
\end{equation}
As \cref{prop:1} will show, \bl{\eqref{eqn:opt}} leads to a generalized eigenvalue problem. \bl{Surprisingly, this} reduces to a Rayleigh quotient maximization problem \bl{since} the shift between $\Delta_{l_1}$ and $\Delta_{l_2}$ ($\l_1\neq l_2$) renders the corresponding noise vectors independent.
\begin{prop}\label{prop:1}
    Suppose $\bm{n}=[n_1,n_2,\cdots,n_N]^T$ where all entries are independent and identically distributed with zero mean and variance $\sigma^2$.
    The solution to \eqref{eqn:opt} is \bl{the largest eigenvector of matrix} $\mathbf{A}\in\mathbb{C}^{L\times L}$, $\mathbf{w}_o=\mathbf{u}_1(\mathbf{A})$, where
    \begin{equation*}
        [\mathbf{A}]_{m,m'}=
        \begin{cases}
            \frac{1}{(m-m')j}(e^{j(m-m')u_h}-e^{j(m-m')u_l} & m\neq m' \\
            u_{h}-u_{l} & m=m'
        \end{cases}.
    \end{equation*}
\end{prop}
\begin{proof}
    Note that $\mathbf{w}^H\mathbf{a}(u)$ is a scalar, \rd{$\|\mathbf{a}_{\mathbb{S}_b}(u)\mathbf{a}_{\Delta_1}^T(u)\|_F^2$ is a constant due to all entries in vectors} 
    $\mathbf{a}_{\mathbb{S}_b}(u)$ and $\mathbf{a}_{\Delta_1}(u)$ 
    \rd{having constant-modulus}, and $x$ is deterministic. Hence, $\max_{\mathbf{w}\in\mathbb{C}^L}\text{SNR}(u)=\max_{\mathbf{w}\in\mathbb{C}^L}\frac{\lvert\mathbf{w}^H\mathbf{a}(u)\rvert^2}{\mathbb{E}(\lVert\mathbf{\Bar{N}}(\mathbf{w})\rVert_F^2)}$. Then we can write 
    \begin{equation}\label{eqn:opt2}
        \begin{aligned}
            \mathbf{w}_o&=\mathop{\arg\max}_{\mathbf{w}\in\mathbb{C}^{L}}\int_{\mathbf{U}}\frac{\lvert\mathbf{w}^H\mathbf{a}(u)\rvert^2}{\mathbb{E}(\lVert\mathbf{\Bar{N}}(\mathbf{w})\rVert_F^2)} \,du\\
            &=\mathop{\arg\max}_{\mathbf{w}\in\mathbb{C}^{L}}\frac{\mathbf{w}^H(\int_{\mathbf{U}}\mathbf{a}(u)\mathbf{a}^H(u)\,du)\mathbf{w}}{\mathbb{E}(\lVert\mathbf{\Bar{N}}(\mathbf{w})\rVert_F^2)}.
        \end{aligned}
    \end{equation}
    Let $\mathbf{A}(u)=\mathbf{a}(u)\mathbf{a}^H(u)$ and $\mathbf{A}=\int_{\mathbf{U}}\mathbf{a}(u)\mathbf{a}^H(u)\,du$. \bl{Then}
    \begin{equation*}
        \begin{aligned}
            [\mathbf{A}]_{m,m'}&=\int_{\mathbf{U}}[\mathbf{A}(u)]_{m,m'}\,du\\
            &=\begin{cases}
                \int_{u_l}^{u_h}e^{j(m-m')u}\,du & m\neq m' \\
                \int_{u_l}^{u_h}1\,du & m = m'
            \end{cases}\\
            &=\begin{cases}
                \frac{1}{(m-m')j}(e^{j(m-m')\pi\sin\theta_{h}}-e^{j(m-m')\pi\sin\theta_{l}}) & m\neq m' \\
                \pi\sin\theta_{h}-\pi\sin\theta_{l} & m=m'.
            \end{cases}
        \end{aligned}
    \end{equation*}
    Let $\mathbf{v}_i=vec(\mathbf{N}_i)$, denote $\mathbf{V}=[\mathbf{v}_1,\mathbf{v}_2,\cdots,\mathbf{v}_L]$, then we have:
    \begin{equation*}
        \begin{aligned}
            \mathbb{E}(\lVert\mathbf{\Bar{N}}(\mathbf{w})\rVert_F^2)
            &=\mathbb{E}(\lVert\sum_{i=1}^L\mathbf{N}_i^{(s)}\cdot w_i^*\rVert_F^2)\\
            &=\mathbb{E}(\lVert\mathbf{V}\mathbf{w}^*\rVert_2^2)\\
            &=\mathbf{w}^T\underbrace{\mathbb{E}(\mathbf{V}^H\mathbf{V})}_{\bl{\triangleq \mathbf{P}}}\mathbf{w}^*.
        \end{aligned}
    \end{equation*}
    We now show that $\mathbf{P}$ is a scaled identity matrix due to the shift between $\Delta_i$ and $\Delta_{i'}$ ($i\neq i'$).
    Define $N_s\triangleq N_bP$ and let $\mathbf{v}_1=[n_{q_1},n_{q_2},\cdots,n_{q_{N_s}}]^T\in\mathbb{C}^{N_s}$. Due to the shift between different $\Delta_l$ in \cref{eqn:sequence of Delta} we have:
    \begin{align*}
        &\mathbf{v}_i=[n_{q_1+(i-1)},n_{q_2+(i-1)},\cdots,n_{q_{N_s}+(i-1)}]^T\\
        &\mathbf{v}_{i'}=[n_{q_1+(i'-1)},n_{q_2+(i'-1)},\cdots,n_{q_{N_s}+(i'-1)}]^T\\
        &\mathbb{E}(\mathbf{v}_i^H\mathbf{v}_{i'})=\sum_{r=1}^{N_s}\mathbb{E}(n_{q_r+(i-1)}^*n_{q_r+(i'-1)}).
    \end{align*}
    Note that both $q_r+(i-1), q_r+(i'-1)\in\{1,2,\cdots,N\}$, if $i\neq i'$, then $q_r+(i-1)\neq q_r+(i'-1)$ and therefore $n_{q_r+(i-1)}$ is independent of $n_{q_r+(j-1)}$ $\forall r\in\{1,2,\cdots,N_s\}$. Thus, $\mathbb{E}(\mathbf{v}_i^H\mathbf{v}_{i'})=0$ for $i\neq i'$. 
    
    Using $\mathbb{E}(\mathbf{v}_i^H\mathbf{v}_j)=0$ for $i\neq j$ and $\mathbb{E}(\mathbf{v}_i^H\mathbf{v}_i)=N_s\sigma^2$, we have $\mathbf{P}=N_s\sigma^2\cdot\mathbf{I}$. Then \cref{eqn:opt2} becomes:
    \begin{equation}
        \begin{aligned}
            \mathbf{w}_o&=\mathop{\arg\max}_{\mathbf{w}\in\mathbb{C}^{L}}\frac{\mathbf{w}^H\mathbf{A}\mathbf{w}}{\mathbf{w}^T\mathbf{P}\mathbf{w}^*}\\
            &=\mathop{\arg\max}_{\mathbf{w}\in\mathbb{C}^{L}}\frac{\mathbf{w}^H\mathbf{A}\mathbf{w}}{\mathbf{w}^H\mathbf{w}}\\
            &=\mathop{\arg\max}_{\mathbf{w}\in\mathbb{C}^{L}}R(\mathbf{A},\mathbf{w}),
        \end{aligned}
    \end{equation}
    where $R(\mathbf{A},\mathbf{w})$ is the Rayleigh quotient \cite{horn2012matrix} for $\mathbf{A}$ and $\mathbf{w}$. \bl{Since} $R(\mathbf{A},\mathbf{w})\leq\lambda_{max}(\mathbf{A})$, equality holds if and only if $\mathbf{w}=\mathbf{u}_1(\mathbf{A})$.
\end{proof}
\cref{fig.bp} \rd{shows an example beampattern corresponding to the optimal shift-domain weights in \eqref{eqn:opt}} for $\mathbf{\Theta}=(10^{\circ},40^{\circ})$ and $L=11$. Note that $\mathbf{w}_o$ can be computed offline once $\mathbf{\Theta}$ is given.
\begin{figure}[htbp]
\centerline{\includegraphics[width=0.45\textwidth]{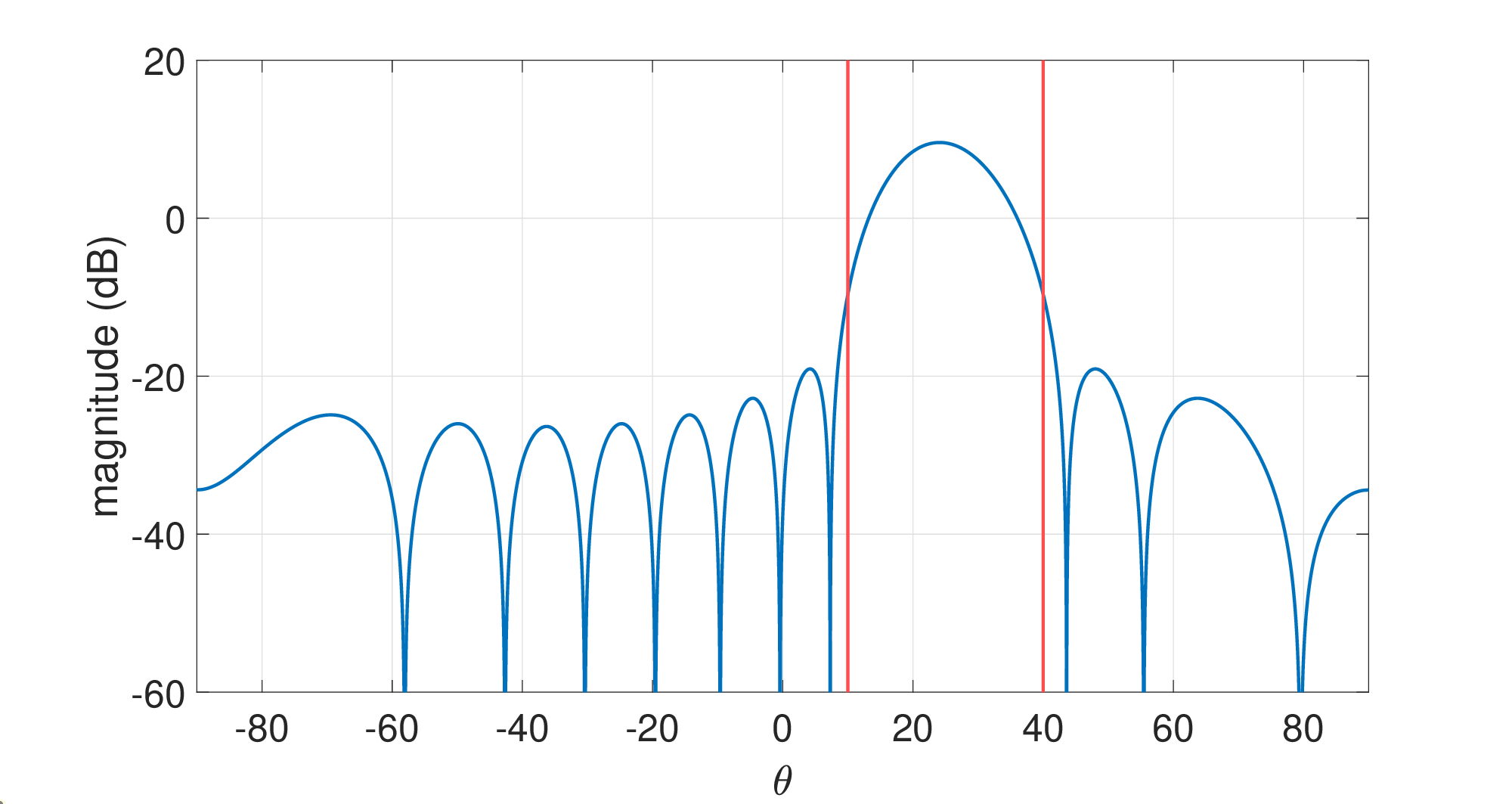}}
\caption{Beampattern of optimal shift-domain weights in \eqref{eqn:opt}. The region of interest in the mainlobe \rd{enjoys a beamforming gain} (red lines indicate $\theta_l,\theta_h$).}
\label{fig.bp}
\end{figure}

\section{Simulation Results}\label{sec:sim}
In our first simulation, we consider $K=3$ sources and $N=89$ sensors. 
The true DOAs are $\bm{\theta}=[20^{\circ}, 25^{\circ},30^{\circ}]$, the source/target signals are $\bm{x}=\frac{1+j}{\sqrt{2}}[1,1,1]^T$, and the noise follows a circularly-symmetric complex normal distribution $\bm{n}\sim\mathcal{CN}(0,\sigma^2\mathbf{I}_{N\times N})$. We vary \bl{noise level} $\sigma$ from 0.1 to 10, which corresponds to an SNR $\triangleq 20\cdot\log_{10}\frac{\min_k(\lvert x_k\rvert)}{\sigma}$ ranging from $20$ dB to $-20$ dB. 
We compute the empirical RMSE per source (in degrees) over 200 Monte Carlo trials \bl{as} $\text{RMSE}=\sqrt{\frac{1}{200K}\sum_{i=1}^{200}\lVert\bm{\hat{\theta}_i}-\bm{\theta}\rVert_2^2}$ where $\bm{\hat{\theta}_i}=[\hat{\theta}_1^{(i)},\hat{\theta}_2^{(i)},\hat{\theta}_3^{(i)}]^T$ is the DOA estimate of the $i$-th \rd{trial}. Spectral MUSIC \cite{schmidt1986multiple} with a grid size of 20000 is used for DOA estimation.

We compare four \rd{spatial smoothing} approaches \rd{with the following parameter values:}
\begin{align*}
    (i)\enspace&\mathbb{S}_{b}=\{0,1,\cdots,7,15,23,31,39,47,55,63,71\},\\
    &\Delta_1=\{0,1,\cdots,7\}, L=11;\\
    (ii)\enspace&\mathbb{S}_{b}=\{0,1,\cdots,8,17,26,35,44,53,62,71,80\},\\
    &\Delta_1=\{0,1,\cdots,8\}, L=1;\\
    (iii)\enspace&\mathbb{S}_{b}=\{0,1,\cdots,80\},\Delta_1=\{0,1,\cdots,8\}, L=1;\\
    (iv)\enspace&\mathbb{S}_{b}=\{0,1,\cdots,15\},\Delta_1=\{0,1,\cdots,73\}, L=1.
\end{align*}
\rd{These cases correspond to} ($i$) $S^3$ with shift-domain beamforming (\rd{nested basic sub-array with} $N_b=16$ sensors, $P=8$ shifts of the basic sub-array per spatially smoothed measurement matrix and $L=11$ shift-domain beamforming weights chosen according to \cref{prop:1} with $\mathbf{\Theta}=(10^{\circ},40^{\circ})$---see \cref{fig.bp}); ($ii$) $S^3$ without shift-domain beamforming (\rd{nested basic sub-array with $N_b=17$, $L=1$} using all $N$ measurements); ($iii$) ULA1 
\rd{with} the same \rd{basic sub-array} aperture as the nested array in ($ii$); and ($iv$) ULA2 \rd{with} the same number of \rd{basic sub-array} sensors as the nested array in \rd{($i$)}.

\begin{figure}[htbp]
\centerline{\includegraphics[width=0.5\textwidth]{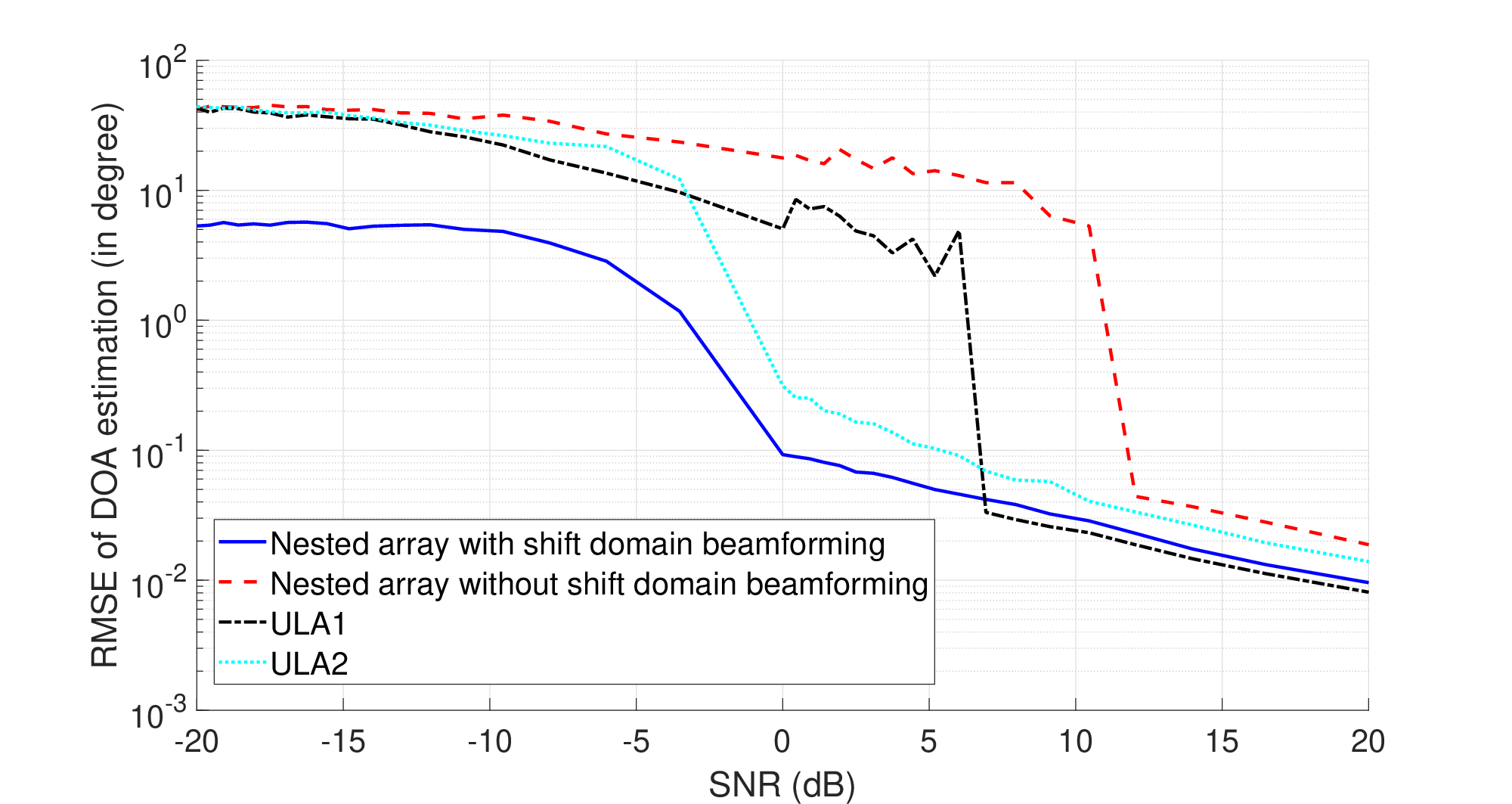}}\vspace{-0.2cm}
\caption{Average DOA estimation error versus SNR. \bl{Sparse spatial smoothing employing shift-domain beamforming \ff{(over region $\mathbf{\Theta}=(10^{\circ},40^{\circ})$)} achieves lower error than conventional} ULA \bl{sub-array-based} approaches at low SNR.}
\label{fig.1}
\end{figure}
\cref{fig.1} shows the DOA estimation error as a function of SNR. The performance of sparse spatial smoothing $S^3$ with \bl{shift-domain} beamforming outperforms conventional ULA-based spatial smoothing when SNR is low, while retaining comparable performance when SNR is high. 
\bl{\cref{fig.1} also illustrates the fact that resolution is not merely a function of aperture, but also the effective SNR \cite{sarangi2023superresolution} as the improved beamforming gain of $S^3$ with shift-domain beamforming compensates for its slightly smaller aperture compared to $S^3$ without shift-domain beamforming.} 
\cref{fig.2} shows the MUSIC pseudo-spectrum of both $S^3$ with shift-domain beamforming (top panel) and the two conventional ULA-based spatial smoothing approaches ULA1 (middle panel) and ULA2 (bottom panel) in the case of close DOA separation ($\bm{\theta}=[20^{\circ}, 22^{\circ}, 24^{\circ}]$). \ff{The left column of \cref{fig.2} shows that when DOA separation is close, ULA2 cannot resolve the DOAs due to its limited aperture.} Moreover, when SNR is low (right column of \cref{fig.2}), the \bl{beamforming gain} of $S^3$ due to shift-domain beamforming enables resolving all three sources, when both ULA1 (with slightly larger \bl{basic sub-array} aperture) and ULA2 (with the same number of basic sub-array sensors) fail \ff{due to a spurious peak in the MUSIC pseudo-spectrum (likely due to noise shaping by spatial smoothing) and decreased resolution, respectively}. Note that $S^3$ offers this improved performance at a reduced computational complexity, as we demonstrate next.
\begin{figure}[htbp]
\includegraphics[width=0.241\textwidth]{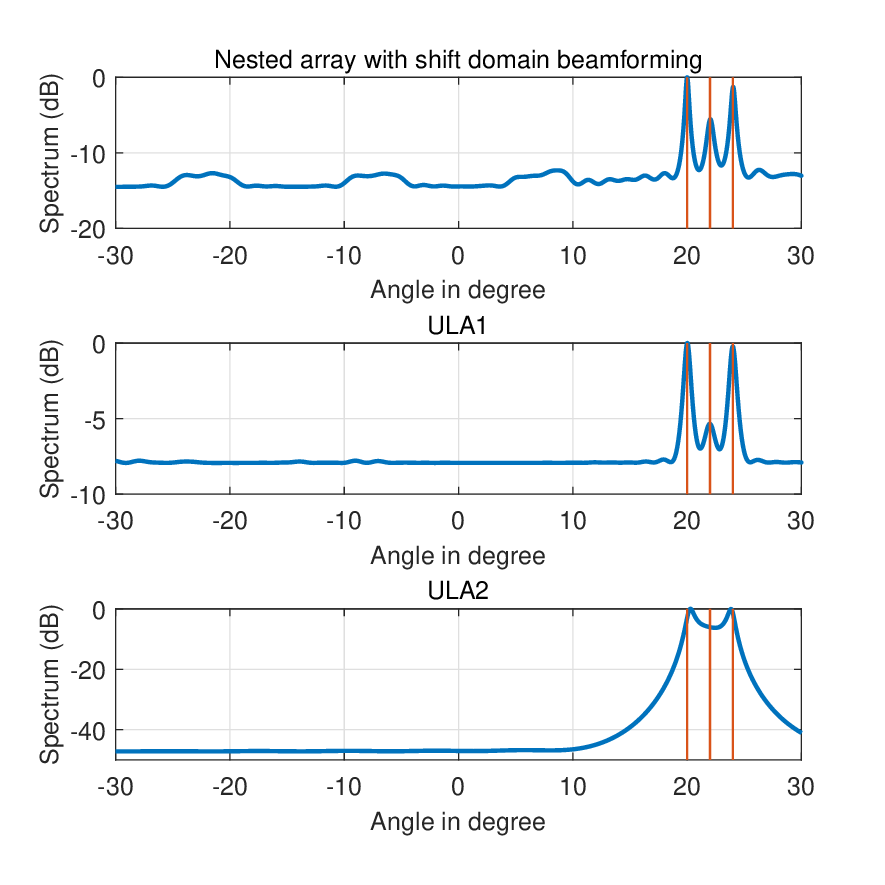}
\includegraphics[width=0.241\textwidth]{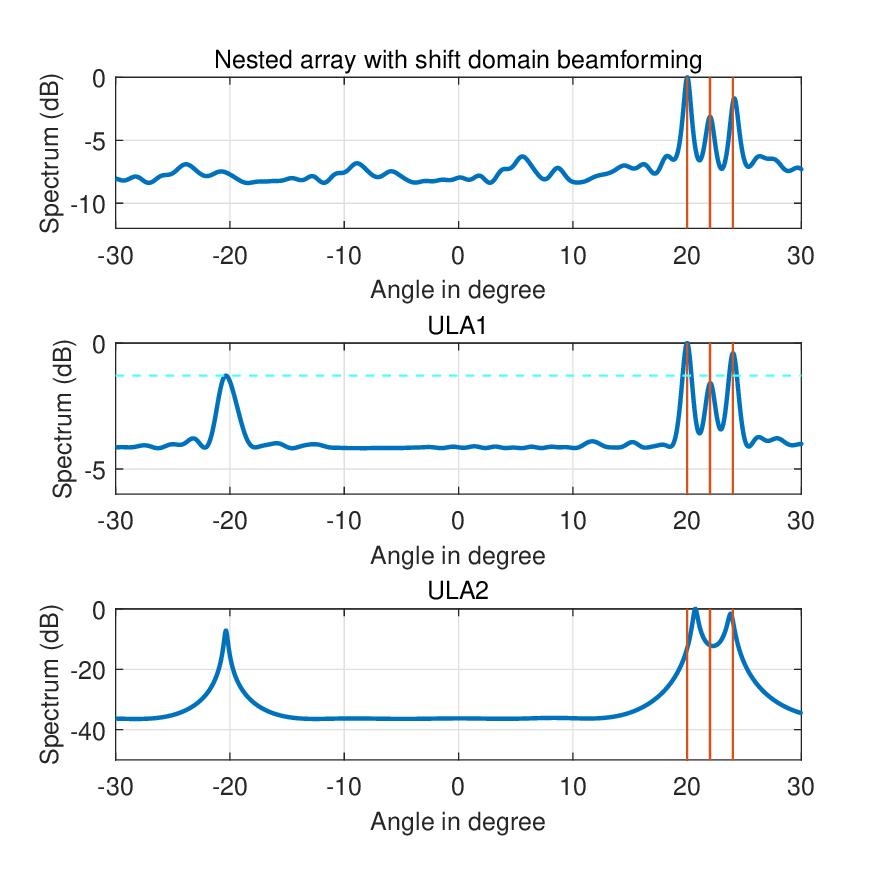}
\caption{MUSIC pseudo-spectrum with (left) SNR = 14 dB and (right) SNR = 0 dB. $S^3$ with shift-domain beamforming (top) provides improved resolution compared to conventional ULA-based spatial smoothing (middle and bottom), especially at low SNR.}
\label{fig.2}
\end{figure}

In our second simulation, we set $N_b=2P$ for the nested array 
$N=Aper(\mathbb{S}_b)+L+P-1$ and vary $P$ from 5 to 15 \ff{(in the previous simulation, $P\!=\!8$)}. 
Other parameters remain unchanged. We compare four approaches with the following spatial smoothing parameters:
\begin{align*}
    (i)\enspace&\mathbb{S}_{b}=\{0,1\cdots,P-1, 2P-1,\rd{3P-1},\cdots, \rd{(P+1)P}-1\},\\
    &\Delta_1=\{0,1,\cdots,P-1\}, L=11;\\
    (ii)\enspace&\mathbb{S}_{b}=\{0,1,\cdots,P-1, 2P-1,\rd{3P-1},\cdots, \rd{(P+1)P}-1\},\\
    &\Delta_1=\{0,1,\cdots,P+9\}, L=1;\\
    (iii)\enspace&\mathbb{S}_{b}=\{0,1,\cdots,\rd{(P+1)P}-1\},\Delta_1=\{0,1,\cdots,P+9\}, L=1;\\
    (iv)\enspace&\mathbb{S}_{b}=\{0,1,\cdots,2P-1\},\Delta_1=\{0,1,\cdots,P^2+9\}, L=1.
\end{align*}
\rd{These cases correspond to} ($i$) $S^3$ with shift-domain beamforming; ($ii$) $S^3$ without shift-domain beamforming
; ($iii$) ULA1 \rd{with} the same \rd{basic sub-array} aperture as the nested array in ($i$); and ($iv$) ULA2 \rd{with} the same number of \rd{basic sub-array} sensors as the nested array in ($i$).

\cref{fig.3} shows the running time in all \bl{four} cases---including constructing the spatially smoothed measurement matrix and applying SVD---averaged over 200 Monte Carlo trials. Clearly, $S^3$ (both with \bl{and without shift-domain} beamforming) achieves much lower computational complexity than conventional ULA-based spatial smoothing. The size of the corresponding spatially smoothed measurement matrices are \rd{($i$)} $2P\times P$; \rd{($ii$)} $2P\times(P+10)$; \rd{($iii$)} $(P^2+P)\times(P+10)$; and \rd{($iv$)} $2P\times(P^2+10)$. \bl{Note that in absence of prior knowledge of the DOAs, we can still employ $S^3$} without shift-domain beamforming to achieve \bl{high-resolution DOA estimation over the whole angular region at} reduced computational complexity.
\begin{figure}[htbp]
\centerline{\includegraphics[width=0.45\textwidth]{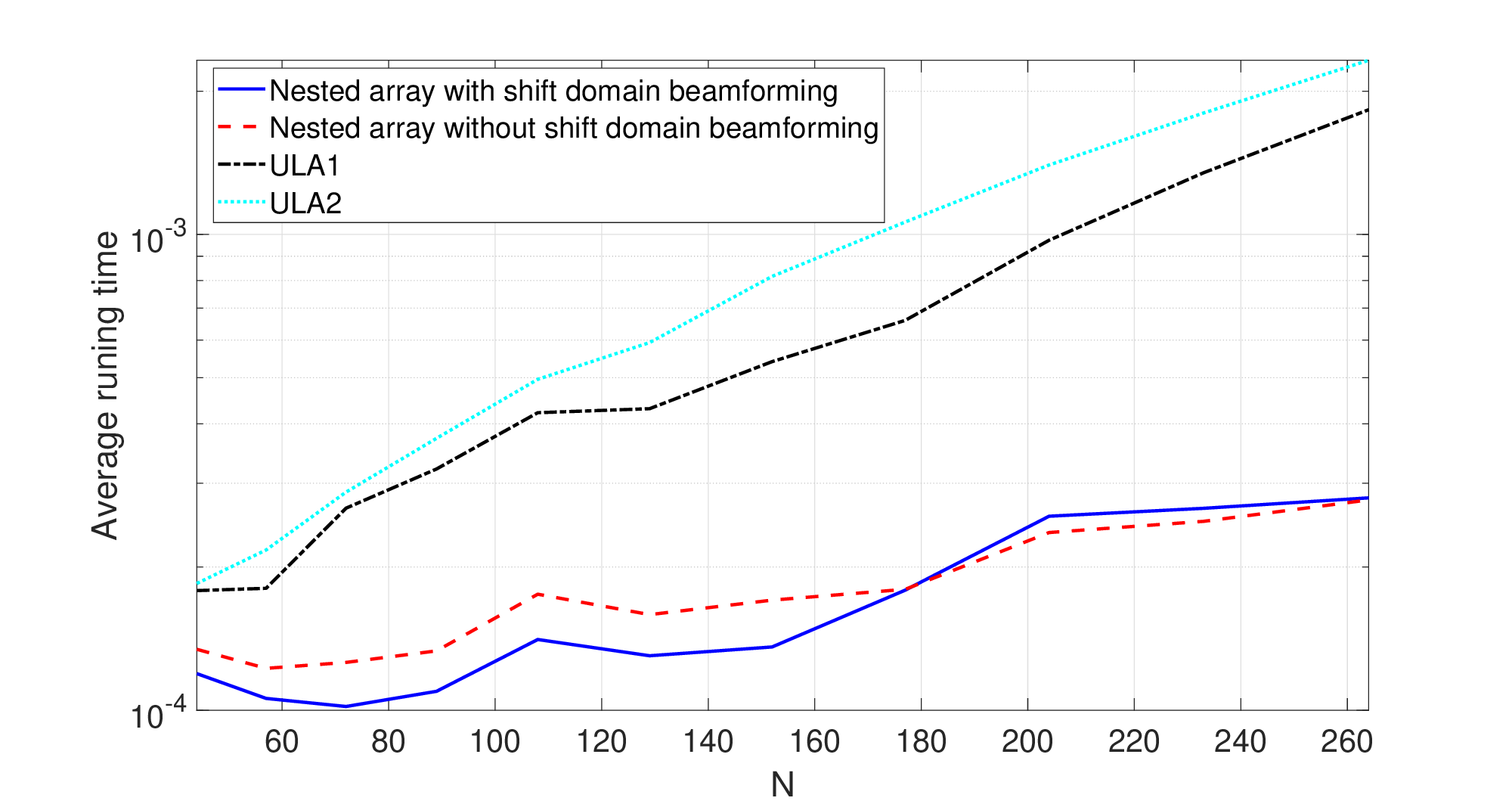}}
\caption{Average running (elapsed real) time versus number of physical sensors $N$. \bl{Sparse spatial smoothing} 
achieves \bl{significantly} lower computational complexity than ULA-based approaches.}
\label{fig.3}
\end{figure}

\section{Conclusions}
This paper \bl{considered} the problem of single-snapshot high-resolution DOA estimation. \bl{Specifically, we showed that the computational complexity of subspace methods employing spatial smoothing on an $N$-antenna ULA can be reduced by a factor of $\sqrt{N}$ by utilizing suitable sparse (instead of ULA) sub-arrays in the spatial smoothing step. This reduction in} complexity \bl{is achieved} while retaining the aperture and identifiability of conventional spatial smoothing \bl{using} ULA sub-arrays \bl{when the number of targets is $\mathcal{O}(\sqrt{N})$}. We also presented a novel beamforming \bl{approach to} sparse array-based spatial smoothing, \bl{where multiple spatially smoothed matrices are linearly combined} to improve the effective SNR in a desired angular region. Simulations \bl{demonstrated} that this \bl{so-called shift-domain beamforming method can improve resolution by appropriate beamforming weight design while} 
 achieving significantly lower computational complexity than conventional ULA-based spatial smoothing. Future work will explore \bl{noise-}robust shift-domain beamforming weight designs \bl{and sub-array choices for sparse spatial smoothing}.
\bibliographystyle{IEEEtran}
\bibliography{ref}

\begin{thebibliography}{10}
\providecommand{\url}[1]{#1}
\csname url@samestyle\endcsname
\providecommand{\newblock}{\relax}
\providecommand{\bibinfo}[2]{#2}
\providecommand{\BIBentrySTDinterwordspacing}{\spaceskip=0pt\relax}
\providecommand{\BIBentryALTinterwordstretchfactor}{4}
\providecommand{\BIBentryALTinterwordspacing}{\spaceskip=\fontdimen2\font plus
\BIBentryALTinterwordstretchfactor\fontdimen3\font minus
  \fontdimen4\font\relax}
\providecommand{\BIBforeignlanguage}[2]{{%
\expandafter\ifx\csname l@#1\endcsname\relax
\typeout{** WARNING: IEEEtran.bst: No hyphenation pattern has been}%
\typeout{** loaded for the language `#1'. Using the pattern for}%
\typeout{** the default language instead.}%
\else
\language=\csname l@#1\endcsname
\fi
#2}}
\providecommand{\BIBdecl}{\relax}
\BIBdecl

\bibitem{patole2017automotive}
S.~M. Patole, M.~Torlak, D.~Wang, and M.~Ali, ``Automotive radars: A review of
  signal processing techniques,'' \emph{IEEE Signal Processing Magazine},
  vol.~34, no.~2, pp. 22--35, 2017.

\bibitem{zhang2021overview}
J.~A. Zhang, F.~Liu, C.~Masouros, R.~W. Heath, Z.~Feng, L.~Zheng, and
  A.~Petropulu, ``An overview of signal processing techniques for joint
  communication and radar sensing,'' \emph{IEEE Journal of Selected Topics in
  Signal Processing}, vol.~15, no.~6, pp. 1295--1315, 2021.

\bibitem{sun20214d}
S.~Sun and Y.~D. Zhang, ``{4D} automotive radar sensing for autonomous
  vehicles: A sparsity-oriented approach,'' \emph{IEEE Journal of Selected
  Topics in Signal Processing}, vol.~15, no.~4, pp. 879--891, 2021.

\bibitem{odendaal1994two}
J.~Odendaal, E.~Barnard, and C.~Pistorius, ``Two-dimensional superresolution
  radar imaging using the {MUSIC} algorithm,'' \emph{IEEE Transactions on
  Antennas and Propagation}, vol.~42, no.~10, pp. 1386--1391, 1994.

\bibitem{pillai1989forward}
S.~U. Pillai and B.~H. Kwon, ``Forward/backward spatial smoothing techniques
  for coherent signal identification,'' \emph{IEEE Transactions on Acoustics,
  Speech, and Signal Processing}, vol.~37, no.~1, pp. 8--15, 1989.

\bibitem{schmidt1986multiple}
R.~Schmidt, ``Multiple emitter location and signal parameter estimation,''
  \emph{IEEE transactions on antennas and propagation}, vol.~34, no.~3, pp.
  276--280, 1986.

\bibitem{roy1989esprit}
R.~Roy and T.~Kailath, ``{ESPRIT}-estimation of signal parameters via
  rotational invariance techniques,'' \emph{IEEE Transactions on acoustics,
  speech, and signal processing}, vol.~37, no.~7, pp. 984--995, 1989.

\bibitem{sun2020mimo}
S.~Sun, A.~P. Petropulu, and H.~V. Poor, ``{MIMO} radar for advanced
  driver-assistance systems and autonomous driving: Advantages and
  challenges,'' \emph{IEEE Signal Processing Magazine}, vol.~37, no.~4, pp.
  98--117, 2020.

\bibitem{roldan2023low}
I.~Roldan, L.~Lamberti, F.~Fioranelli, and A.~Yarovoy, ``Low complexity
  single-snapshot {DoA} estimation via {Bayesian} compressive sensing,'' in
  \emph{2023 IEEE Radar Conference (RadarConf23)}.\hskip 1em plus 0.5em minus
  0.4em\relax IEEE, 2023, pp. 1--6.

\bibitem{wang2017coarrays}
M.~Wang and A.~Nehorai, ``Coarrays, {MUSIC}, and the {C}ram\'{e}r-{R}ao
  bound,'' \emph{IEEE Transactions on Signal Processing}, vol.~65, no.~4, pp.
  933--946, Feb 2017.

\bibitem{sarangi2023superresolution}
P.~Sarangi, M.~C. H\"{u}c\"{u}meno\u{g}lu, R.~Rajam\"{a}ki, and P.~Pal,
  ``Super-resolution with sparse arrays: A non-asymptotic analysis of
  spatio-temporal trade-offs,'' \emph{IEEE Transactions on Signal Processing},
  pp. 1--14, 2023.

\bibitem{liu2016supernested}
C.-L. Liu and P.~P. Vaidyanathan, ``Super nested arrays: Linear sparse arrays
  with reduced mutual coupling -- {P}art {I}: {F}undamentals,'' \emph{IEEE
  Transactions on Signal Processing}, vol.~64, no.~15, pp. 3997--4012, Aug
  2016.

\bibitem{pal2010nested}
P.~Pal and P.~P. Vaidyanathan, ``Nested arrays: A novel approach to array
  processing with enhanced degrees of freedom,'' \emph{IEEE Transactions on
  Signal Processing}, vol.~58, no.~8, pp. 4167--4181, 2010.

\bibitem{li2007mimo}
J.~Li and P.~Stoica, ``{MIMO} radar with colocated antennas,'' \emph{IEEE
  signal processing magazine}, vol.~24, no.~5, pp. 106--114, 2007.

\bibitem{shan1985spatial}
T.-J. Shan, M.~Wax, and T.~Kailath, ``On spatial smoothing for
  direction-of-arrival estimation of coherent signals,'' \emph{IEEE
  Transactions on Acoustics, Speech, and Signal Processing}, vol.~33, no.~4,
  pp. 806--811, 1985.

\bibitem{bu2023harnessing}
Y.~Bu, R.~Rajam\"{a}ki, P.~Sarangi, and P.~Pal, ``Harnessing holes for spatial
  smoothing with applications in automotive radar,'' in \emph{57th Asilomar
  Conference on Signals, Systems, and Computers}, 2023 (accepted for
  publication).

\bibitem{golub2013matrix}
G.~H. Golub and C.~F. Van~Loan, \emph{Matrix computations}.\hskip 1em plus
  0.5em minus 0.4em\relax JHU press, 2013.

\bibitem{horn2012matrix}
R.~A. Horn and C.~R. Johnson, \emph{Matrix analysis}.\hskip 1em plus 0.5em
  minus 0.4em\relax Cambridge university press, 2012.

\end{thebibliography}

\end{document}